\documentclass[11pt,reqno]{amsart}
\usepackage[verbose=true]{microtype}
\usepackage{amsmath}
\usepackage{amssymb}
\usepackage{amsthm}
\usepackage{bbm}
\usepackage{color}
\usepackage[colorlinks,linkcolor=black,citecolor=black,linktocpage,hypertexnames=true,pdfpagelabels]{hyperref}
\usepackage[capitalise]{cleveref}
\usepackage[all]{xy}
\usepackage{subcaption}
\usepackage{array}
\usepackage{enumerate}
\usepackage{enumitem}
\usepackage{tikz}
\usepackage{relsize}
\usepackage{diagbox}
\usepackage{tikz-cd} 
\usepackage{mdframed}
\usepackage{dsfont}
 \usepackage[foot]{amsaddr}
\newtheorem{theorem}{Theorem}
\newtheorem{lemma}[theorem]{Lemma}
\newtheorem{corollary}[theorem]{Corollary}
\newtheorem{proposition}[theorem]{Proposition}

\newtheorem{definition}[theorem]{Definition}
\theoremstyle{definition}

\newtheorem{example}[theorem]{Example}
\newtheorem{remark}[theorem]{Remark}
\newtheorem{construction}[theorem]{Construction}

\renewcommand{\epsilon}{\varepsilon}
\newcommand{\N}{\mathbb{N}}
\newcommand{\C}{\mathbb{C}}
\newcommand{\R}{\mathbb{R}}

\newcommand{\be}{\begin{eqnarray}}
\newcommand{\ee}{\end{eqnarray}}
\newcommand{\ben}{\begin{enumerate}}
\newcommand{\een}{\end{enumerate}}

\newcommand{\ba}{\begin{array}}
\newcommand{\ea}{\end{array}}

\newcommand{\mc}{\mathcal}

\newcommand{\Her}{\mathrm{Her}}

\usepackage{xfrac}

\makeatletter
\def\p@subsection{}
\def\p@subsubsection{}
\makeatother

\usepackage[textsize=footnotesize]{todonotes}

\let\originalleft\left
\let\originalright\right
\renewcommand{\left}{\mathopen{}\mathclose\bgroup\originalleft}
\renewcommand{\right}{\aftergroup\egroup\originalright}

\newcommand\xqed[1]{%
  \leavevmode\unskip\penalty9999 \hbox{}\nobreak\hfill
  \quad\hbox{#1}}
\newcommand\demo{\xqed{$\triangle$}}

\begin{document} 

\title{
Quantum magic squares: dilations and  their limitations
}

\author{Gemma De las Cuevas}
\address{Institute for Theoretical Physics, Technikerstr.\ 21a,  A-6020 Innsbruck, Austria}

\author{Tom Drescher}
\address{Department of Mathematics, Technikerstr.\ 13,  A-6020 Innsbruck, Austria}

\author{Tim Netzer}
\address{Department of Mathematics, Technikerstr.\ 13,  A-6020 Innsbruck, Austria}

\date{\today}

\begin{abstract}
Quantum permutation matrices and quantum magic squares are generalizations of permutation matrices and magic squares, where the entries are no longer numbers but elements from arbitrary (non-commutative) algebras. The famous Birkhoff--von Neumann Theorem characterizes magic squares as convex combinations of permutation matrices. In the non-commutative case, the corresponding question is: Does every quantum magic square belong to the matrix convex hull of quantum permutation matrices? That is, does every quantum magic square dilate to a quantum permutation matrix? Here we show that this is false even in the simplest non-commutative case. We also classify the quantum magic squares that dilate to a quantum permutation matrix with commuting entries, and prove a quantitative lower bound on the diameter of this set. Finally, we conclude that not all Arveson extreme points of the free spectrahedron of quantum magic squares are quantum permutation matrices.
\end{abstract}

\maketitle

\tableofcontents

\section{Introduction}
\label{sec:intro}

A magic square is a square of nonnegative numbers with the property that every row and column sums to the same value. (Often it is additionally required that the diagonals also sum to this value, but we will not require this here).
If the number they sum to is one, magic squares are also called doubly stochastic matrices.  Throughout this work we will use the words magic square and doubly stochastic matrix interchangeably, i.e.\ we always assume  that all rows and columns of our matrices sum to one.

The Birkhoff--von Neumann Theorem  is a central result in this context, which says that  every magic square is a convex combination of permutation matrices---see for example \cite{Ba02}. Furthermore, the permutation matrices constitute precisely  the vertices of the polytope of all magic squares.  This result is a beautiful and useful characterization of magic squares, with applications in the theory of Markov chains, matrix approximations, the graph isomorphism problem, market design and many more---see for example \cite{Bru87, Bud09}.

Several generalizations of magic squares and their corresponding Birkhoff--von Neumann Theorem have been proposed in the literature. 
One such generalization are unital quantum channels, that is, completely positive and trace preserving linear maps which map the identity to itself. 
It is easy to see that (the diagonal elements of) unital quantum channels that map diagonal matrices to diagonal matrices correspond to magic squares.
A version of the Birkhoff--von Neumann  Theorem for unital quantum channels has been shown to fail in  \cite{Lan93}, and an asymptotic version thereof has also been shown to fail  in  \cite{Haa11}; see also \cite{Me09b} for a detailed analysis.

In this work we consider a different generalization of magic squares, which is  inspired by non-commutative semialgebraic geometry and free convexity (see e.g.\  \cite{Ne19}). 
Namely, we define magic squares over an algebra $\mc{A}$ as matrices containing positive elements from $\mc{A}$, and such that every row and column sums to the identity element in $\mc{A}$. We call such matrices \emph{quantum magic squares}.
If, additionally, every element is a projector, then the magic square is called a \emph{quantum permutation matrix} (also called magic unitary in  \cite{Lu17, Pa19}).
For readers familiar with quantum information theory, 
a magic square over the algebra  of  complex matrices of size $s\times s$,  ${\rm Mat}_s(\mathbb{C})$,  
contains a positive operator valued measurement (POVM) in each row and in  each column, 
and a quantum permutation matrix contains a projective measurement in each row and column. 
If the algebra $\mc A$ is simply the complex numbers $\mathbb{C}$, 
then quantum  magic squares and quantum permutation matrices coincide 
 with magic squares and permutation matrices, respectively.
We refer to this special case as the ``classical" case,  intuitively because quantum behaviour is associated to non-commutativity, and thereby to non-commutative algebras such as ${\rm Mat}_s(\mathbb{C})$.
Note that quantum permutation matrices have been studied in several contexts: as $*$-representations of the quantum permutation group \cite{Ba07}, 
 to characterize quantum isomorphism of graphs \cite{At19} 
and to represent perfect quantum strategies for non-local games \cite{Lu17, Pa19}.

To be more specific, in this work we consider the free (i.e.\ non-commutative) semialgebraic sets of all magic squares  and of all quantum permutation matrices over matrix algebras. 
The first 
level of these free semialgebraic sets corresponds to the classical case: we recover the sets of magic squares and of permutation matrices, respectively.
We then investigate 
 the \emph{matrix convex hull} 
of the quantum permutation  matrices. 
This is a special notion of convexity which is well-adapted to non-commutative sets, as it relates the different matrix levels by isometric conjugations. 
In a more operator theoretic formulation, we investigate which quantum magic squares 
admit a \emph{dilation} to a quantum permutation matrix. 
While a single POVM  can be dilated to a projective measurement via Stinespring's Dilation Theorem \cite{Pau02}, this does not imply that a quantum magic square  
can be dilated to a quantum permutation matrix, since the former contains a POVM in every row and every column, and we are asking for a \emph{simultaneous} dilation of all of them. 

To  address these questions, we first completely characterize  the matrix convex hull of all quantum permutation matrices with \emph{commuting} entries (see \cref{ssec:free} for a definition of the matrix convex hull).
Namely, we 
show that a quantum magic square dilates to a quantum permutation matrix with commuting entries if and only if it is \emph{semiclassical} (\cref{thm:nodil}). 
Semiclassical means that 
it can be expressed as a sum of tensor products of permutation matrices and positive semidefinite matrices.     
Throughout this paper, we denote the number of rows (which equals the number of columns) of the quantum magic square by $n$. 
When the algebra $\mc{A}$ is ${\rm Mat}_s(\mathbb{C})$ with $s\geq 2$, 
we  show that 
there exist quantum magic squares of all sizes $n\geq 3$ which are not semiclassical (\cref{cor:nodil}).
We also prove a lower bound on the diameter of the free convex set of semiclassical magic squares (\cref{thm:nodil}).

The matrix convex hull of \emph{all} quantum permutation matrices is much harder to describe. We prove that for every size 
$n\geq 3$ there exist quantum magic squares of size $n$ over ${\rm Mat}_2(\mathbb{C})$ that do not dilate to any quantum permutation matrix (\cref{thm:mconv}).
This means that the matrix convex hull of the set of quantum permutation matrices is  strictly contained in the set of all quantum magic squares. 
We conclude
 that all quantum permutation matrices are Arveson extreme points
(see \cref{def:arveson})
 of the free spectrahedron of quantum magic squares, but that for  $n\geq 3$ not every Arveson extreme point is a quantum permutation matrix.

This paper is organized as follows. \cref{sec:prelim} contains  preliminary material   needed in the course of the paper, including quantum magic squares and permutation matrices,  basic concepts of non-commutative semialgebraic geometry and convexity, and operator systems. 
\cref{sec:main} contains our main  results. 
In the first part (\cref{ssec:good}), 
we give a complete characterization of the quantum magic squares that dilate to a quantum permutation matrix with commuting entries (\cref{thm:nodil}), 
we prove a  lower bound on the diameter of the set of these quantum magic squares, 
and  we conclude that there exist non-semiclassical quantum magic squares of any size $n\geq 3$, 
and non-semiclassical quantum permutation matrices of any size $n\geq 4$ (\cref{cor:nodil}). 
In the second part  (\cref{ssec:bad}), 
we show that for any $n\geq 3$ there exist quantum magic squares of size $n$
that do not dilate to any quantum permutation matrix (\cref{thm:mconv}),
and we conclude that quantum permutation matrices are Arveson extreme points of the set of all quantum magic squares, but not every Arveson extreme point is of that form (\cref{cor:arv}).


\section{Preliminaries}
\label{sec:prelim}

In this section we define and characterize 
quantum magic squares and quantum permutation matrices   (\cref{ssec:qperm}), 
free semialgebraic sets and matrix convex hulls (\cref{ssec:free}),
and  certain operator systems (\cref{ssec:cones}).

\subsection{Quantum magic squares and quantum permutation matrices}
\label{ssec:qperm}

Throughout this paper let $\mc A$ be a $C^*$-algebra. The real vector space of Hermitian elements in $\mc A$ is denoted by $\mc A_{\rm her}$, and the closed convex cone of positive elements in $\mc A_{\rm her}$ by $\mc A^+$. 
We will mostly assume that $\mc A={\rm Mat}_s(\C)$, the $C^*$-algebra of complex matrices of size $s\times s$. 
Instead of ${\rm Mat}_s(\C)_{\rm her}$ we write ${\rm Her}_s(\C)$, 
and instead of $a\in {\rm Mat}_s(\C)^+$ we just write $a\geqslant 0$.

\begin{definition}
Let $n\in\N$ and   $A=\left( a_{ij}\right)_{i,j}\in{\rm Mat}_n(\mc A).$

(i) $A$ is called a \emph{magic square over $\mc A$} 
 if all $a_{ij}\in\mc A^+$ and  all rows and columns of $A$ sum to $1,$ the identity element in $\mc A$. A magic square over some algebra $\mc{A}$ is called a \emph{quantum magic square}. 

(ii) A magic square over $\mc A$ is called a 
\emph{permutation matrix over $\mc A$}  
if all $a_{ij}$ are projectors, i.e.\ fulfill $a_{ij}^2=a_{ij}$. 
A permutation matrix over some $\mc A$ is called a \emph{quantum permutation matrix}.
\end{definition}

\begin{remark}\label{rem:squares}
(i) Every quantum permutation matrix is a unitary  in the sense that $$A^*A=AA^*=I_n\in{\rm Mat}_n(\mc A).$$ 
Although the defining conditions only require 
$$\sum_{j} a_{ij}=\sum_{i}a_{ij}=1$$ 
for all $i,j$, the fact that the $a_{ij}$ are projectors in a $C^*$-algebra implies that  
$$a_{ij}a_{ik}=0=a_{ji}a_{ki}$$ for all $i$ and $j\neq k$.

(ii) In the case $\mc A=\C$,  
quantum magic squares become magic squares 
and quantum permutations become permutation matrices (that is, a matrix with exactly one 1 in every row and column, and 0 elsewhere).
The well-known Birkhoff--von Neumann Theorem states that permutation matrices are precisely the vertices of the polytope of all magic squares, see for example \cite{Ba02, Zie95}. 
In particular, every magic square  
is a convex combination of permutation matrices. 
This holds for all matrix sizes $n\geq 1$.

(iii) For matrix size $n=1$ and any $\mc A$, the only magic square over $\mc A$ is $(1)$. For $n=2$ every magic square over $\mc A$ is of the form 
\begin{equation*}
\left(\begin{array}{cc}a & 1-a \\1-a & a\end{array}\right)=\left(\begin{array}{cc}1 & 0 \\0 & 1\end{array}\right)\otimes a + \left(\begin{array}{cc}0 & 1 \\1 & 0\end{array}\right)\otimes (1-a)
\end{equation*} 
with $a,1-a\in\mc A^+.$ Whether this is a 
convex combination of permutation matrices over $\mc A$ just depends on whether $a$ is a convex combination of projectors in $\mc A$. 
In case that $\mc A={\rm Mat}_s(\C)$ is a matrix algebra, this was first shown to hold in \cite{Choi90} for any such $a$. 
If $\mc A=\mc B(\mc H)$ is the algebra of bounded operators on an infinite-dimensional Hilbert space $\mc H$, this is not true with \emph{finite} convex combinations \cite{Fil67}. 
But from the spectral theorem it follows that it holds up to the closure, 
even over any  $C^*$-algebra. 
So the Birkhoff--von Neumann Theorem holds for matrices of size $2$ over matrix algebras with 
convex combinations, and over arbitrary $C^*$-algebras with 
convex combinations and closures. 
This is known to fail for $n\geq 3$ and we will prove a  
stronger statement below.
\end{remark}

We denote the permutation group of $n$ elements by $S_n$, as customary. 

\begin{definition}
A magic square  $A\in{\rm Mat}_n(\mc A)$
over $\mc{A}$
 is called \emph{semiclassical} if it can be written as 
$$A=\sum_{\pi\in S_n} P_{\pi}\otimes q_\pi,$$ 
where $P_{\pi}\in {\rm Mat}_n(\C)$ are permutation matrices and $q_\pi\in\mc A^+$ (with $\sum_\pi q_{\pi}=1$).
\end{definition}

\begin{remark}
\label{rem:semiclassical}
(i) For $n=1,2$ all quantum magic squares are semiclassical. For $\mc A=\C$ and all $n$, all magic squares over $\mc A$ are semiclassical. We will show in \cref{cor:nodil} that for $n\geq 3$ and $\mc A={\rm Mat}_s(\C)$ with   $s\geq 2$  there exist magic squares over $\mc A$ that are not semiclassical. There even exist non-semiclassical quantum permutation matrices, at least for $n\geq 4$ and $s\geq 2$.

(ii) Let $A=(a_{ij})\in\operatorname{Mat}_n(\mathcal{A})$ be a magic square over $\mathcal{A}$ and let $\rho\colon S_n\to\{1,\dots,n!\}$ be a bijection. We set
\begin{align*}
	B &:= \begin{pmatrix}
	0 & 1 & 0 & 0 & 0\\
	1 & 0 & 0 & 0 & 0\\
	0 & 0 & 0 & 0 & 0\\
	0 & 0 & 0 & 0 & 0\\
	0 & 0 & 0 & 0 & 0
	\end{pmatrix},\quad C_{ij} := \begin{pmatrix}
	0 & 0 & 0 & 0 & 0\\
	0 & 0 & 0 & 0 & 0\\
	0 & 0 & 0 & 0 & 0\\
	0 & 0 & 0 & 0 & E_{ij}\\
	0 & 0 & 0 & E_{ij}^t & 0
	\end{pmatrix},\\
	D_{\pi} &:= \begin{pmatrix}
	0 & -1 & 0 & 0 & 0\\
	-1 & 0 & 0 & 0 & 0\\
	0 & 0 & E_{\rho(\pi)\rho(\pi)} & 0 & 0\\
	0 & 0 & 0 & 0 & -P_{\pi}\\
	0 & 0 & 0 & -P_{\pi}^t & 0
	\end{pmatrix}
\end{align*}
where $E_{kl}$ denotes the canonical basis of either $\operatorname{Mat}_{n}(\C)$ or $\operatorname{Mat}_{n!}(\C)$. Then $A$ is semiclassical if and only if there exist $q_{\pi}\in\mathcal{A}$ for all $\pi\in S_n$ such that
\begin{equation*}
	B\otimes 1_{\mathcal{A}}+\sum_{i,j=1}^{n} C_{ij}\otimes a_{ij}+\sum_{\pi\in S_n}^{} D_{\pi}\otimes q_{\pi}
\end{equation*}
is an element of $\operatorname{Mat}_{n!+2n+2}(\mathcal{A})^+$. Thus, for $\mathcal{A}=\operatorname{Mat}_s(\C)$ checking whether a magic square is semiclassical reduces to a semidefinite program.
\end{remark}

\subsection{Matrix convex hulls}
\label{ssec:free} In this section we consider non-commutative sets, which are sets of matrix tuples, for any size of the matrices. As we will  see, the set of quantum magic squares and of quantum permutation matrices are examples of such sets. 

\begin{definition}
For $n,s\geq 1$ we define 
\begin{align*}
\mc M_{s}^{(n)}&:=\left\{ A\in {\rm Mat}_n\left({\rm Her}_s(\C)\right) \mid A \mbox{  quantum  magic square}\right\} \\
\mc P_{s}^{(n)}&:=\left\{ A\in {\rm Mat}_n\left({\rm Her}_s(\C)\right) \mid A \mbox{ quantum permutation matrix}\right\}\\ 
\mc{CP}_{s}^{(n)}&:=\left\{ A\in \mc P_s^{(n)} \mid \mbox{all entries of $A$ commute}\right\}
\end{align*} and 
$$
\mc M^{(n)}:=\bigcup_{s\geq 1} \mc M_{s}^{(n)} \qquad \mc P^{(n)}:=\bigcup_{s\geq 1} \mc P_{s}^{(n)} \qquad \mc{CP}^{(n)}:=\bigcup_{s\geq 1} \mc{CP}_{s}^{(n)}.
$$
\end{definition}

We have the obvious inclusions 
$$
\mc{CP}^{(n)}\subseteq \mc P^{(n)} \subseteq \mc M^{(n)} .
$$ 

For $n\leq 3$ we have $$\mc{CP}^{(n)}=\mc P^{(n)}.$$ This is clear for $n=1,2$ and for $n=3$ a nice elementary proof can be found in \cite{Lu17}. For $n\geq 4$  $$\mc{CP}^{(n)}\subsetneq \mc P^{(n)}$$ holds, which can easily be seen by taking a 
block diagonal sum of two quantum permutations of size $2$ each.

As customary, we use ${\rm Mat}_{s,t}(\C)$ to denote the set of matrices of size $s\times t$ over the complex numbers.

\begin{definition}\label{def:mconv}
For each $s\geq 1$ let some set $R_s\subseteq {\rm Her}_s(\C)^m$ be given. 
Then $R=\bigcup_{s\geq 1} R_s$ is called \emph{matrix convex} if for all $r,s_i,t\geq 1$, $a^{(i)}\in R_{s_i}$ for $i=1,\ldots, r$ and $v_i\in {\rm Mat}_{s_i,t}(\C)$ with $\sum_{i=1}^r v_i^*v_i=I_t$, 
$$
\sum_{i=1}^r v_i^*a^{(i)}v_i\in R_t
$$ 
holds. Here, compression of the $m$-tuple 
$a^{(i)}$ by $v_i$ is meant entrywise.
That is, we identify $ {\rm Her}_s(\C)^m$ with ${\rm Mat}_{1,m}(\C)\otimes {\rm Her}_s(\C)$ 
so that 
$v_i^* a^{(i)}v_i$ is identified with 
$ ( 1 \otimes v_i^*)a^{(i)} ( I_m \otimes v_i).$
\end{definition}

For more information on matrix convexity, see e.g. \cite{Ef97}.

\begin{remark}\label{rem:mconv}
(i) If $R=\bigcup_{s\geq 1} R_s$ is matrix convex, then each $R_s$ is convex in the 
standard
 sense.  But the converse is generally not true, that is, being matrix convex is generally a stronger requirement than simply being convex at each level $s$.

(ii) 
The condition $\sum_{i=1}^r v_i^*a^{(i)}v_i$ can be written as  $v^*av$ by defining  
$$v=(v_1, \ldots, v_r)^t$$ 
and 
$$
a =\left(\bigoplus_{i=1}^r a_1^{(i)}, \ldots, \bigoplus_{i=1}^r a_m^{(i)}\right) ,
$$
where we have written $a^{(i)} = (a^{(i)}_1,\ldots, a^{(i)}_m)$. 
If the set  $R$ is closed under direct sums, then $a\in R$, and thus for these sets the definition of matrix convexity can be simplified to include this single summand. 
One example of a set closed under direct sums are  free semialgebraic sets, i.e.\ sets defined by non-commutative polynomial inequalities (see for example \cite{Ne19} for more information).

(iii) Intersections of matrix convex sets are matrix convex. It follows that for every set $R=\bigcup_{s\geq 1} R_s$ there exists the matrix convex hull, i.e.\ the smallest matrix convex superset, which is obtained by  adding all compressions of elements from $R$ as in \cref{def:mconv}, 
and which we denote by ${\rm mconv}(R)$.
\end{remark}

\begin{example}
The set $\mc M^{(n)}$ of quantum magic squares is matrix convex. For $n\geq 2$ the sets $\mc{CP}^{(n)}$ and $\mc{P}^{(n)}$ are not. We obtain 
$${\rm mconv}\left(\mc{CP}^{(n)}\right)\subseteq {\rm mconv}(\mc P^{(n)}) \subseteq \mc M^{(n)}.$$ 
For $n\leq 2$ these inclusions are all equalities. For $n=2$ this follows from \cref{rem:squares} (iii) and \cref{rem:mconv} (i). For $n=3$ the inclusion on the left is an equality.
\end{example}

\subsection{Magic cones and operator systems}
\label{ssec:cones}

Let $\mc V^{(n)}\subseteq {\rm Mat}_n(\C)$ be the space of all matrices with constant row and column sums, equipped with the involution of entrywise conjugation. Let $$\mc C^{(n)}\subseteq \mc V_{\rm her}^{(n)}$$ be the  convex cone of real and entrywise nonnegative matrices from $\mc V^{(n)}$, that is, the set of magic squares in the non-normalized sense.
 We call $\mc C^{(n)}$ the \emph{magic cone}.

\begin{lemma}\label{lem:cone}
(i) $\mc V^{(n)}$ is spanned over $\C$ by the permutation matrices, $\mc V^{(n)}_{\rm her}$ is spanned over $\R$ by the permutation matrices, and $$\dim_{\C}\left(\mc V^{(n)}\right)=\dim_{\R}\left(\mc V^{(n)}_{\rm her}\right)=(n-1)^2+1.$$

(ii) $\mc C^{(n)}$ is a  
polyhedral cone that does not contain a full line   in $\mc V^{(n)}_{\rm her},$
and the all-ones matrix is an interior point.

(iii) The extreme rays of $\mc C^{(n)}$ are precisely the rays spanned by the permutation matrices, of which there are $n!$ many.  For $n\geq 3$, $\mc C^{(n)}$ has precisely $n^2$ facets, defined by the inequalities $x_{ij}\geq 0$ on the matrix entries, 
and $\mc C^{(n)}$ is not a simplex cone (i.e.\ the number of extreme rays exceeds its dimension). 
\end{lemma}

\begin{proof} All statements are easy observations or straightforward translations of the corresponding well-known facts about the Birkhoff polytope, see for example  \cite{Ba02,Zie95}.
\end{proof}

We now introduce the smallest and the largest operator system over the magic cone $\mc C^{(n)}$, and collect some of their properties. We refer the reader to  \cite{Fr16} for more detailed explanations of these notions.

\begin{definition}(i) For $s\in\N$ we define 
$$\mc S_s^{(n)} := \left\{ \sum_{\pi\in S_n} P_\pi \otimes q_\pi\mid  q_\pi\in {\rm Mat}_s(\C)^+\right\}
$$ 
where $P_{\pi}\in {\rm Mat}_n(\C)$ are the  permutation matrices, 
and understand this convex cone as a subset of the Hermitian part of 
$$
\mc V^{(n)}\otimes {\rm Mat}_s(\C)\subseteq {\rm Mat}_n\left({\rm Mat}_s(\C) \right).
$$  We call 
$$
\mc S^{(n)}:=\bigcup_{s\geq 1} \mc S_s^{(n)}
$$ 
the \emph{smallest magic operator system}.

(ii) For $s\in\N$ we define 
$$
\mc  L_s^{(n)}:=\left\{ A \in (\mc V^{(n)}\otimes {\rm Mat}_s(\C))_{\rm her}\mid \forall v\in\C^s\colon (I_n\otimes v)^*A(I_n\otimes v)\in \mc C^{(n)}\right\}
$$ 
and call   
$$
\mc L^{(n)}:=\bigcup_{s\geq 1} \mc L_s^{(n)}
$$ 
the \emph{largest magic operator system}.
\end{definition}

\begin{remark}\label{rem:mos}
(i) $\mc S^{(n)}$ and $\mc L^{(n)}$ are really the smallest/largest abstract operator systems over $\mc C^{(n)}$, as defined in \cite{Fr16}. The only difference is that we consider $\mc V^{(n)}\otimes {\rm Mat}_s(\C)$ as a subspace of ${\rm Mat}_n\left({\rm Mat}_s(\C)\right)$, instead of identifying it with ${\rm Mat}_s\left(\mc V^{(n)}\right)$.

(ii) Magic squares of size $n$  over ${\rm Mat}_s(\C)$ are precisely the elements from $\mc L_s^{(n)}$ with the additional property that every row and column sums to  $I_s\in{\rm Mat}_s(\C)$.  
Such a quantum magic square is semiclassical if and only if it belongs to $\mc S_s^{(n)}.$

(iii) From \cref{lem:cone} and \cite[Theorem 1]{Hu19} it follows that for $n\leq 2$ we have $\mc S^{(n)}=\mc L^{(n)}$, but $$\mc S_s^{(n)}\subsetneq \mc L_s^{(n)}$$ for all $n\geq 3$ and $s\geq 2$. 
\end{remark}

\section{Dilations of quantum magic squares}

\label{sec:main}

In this section we present the main results of this paper. Namely, we characterize semiclassical quantum magic squares (\cref{ssec:good}), 
and the matrix convex hull of quantum permutations (\cref{ssec:bad}).

\subsection{Characterization of semiclassical  quantum magic squares}
\label{ssec:good}
The following is our main result concerning the characterization of semiclassical quantum magic squares.  

\begin{theorem}\label{thm:nodil} 
%
%
Let $a=(a_{ij})\in\mathcal{M}_s^{(n)}$ be a quantum magic square.
\begin{enumerate}
	\item[(i)] Let $\C^{S_n}$ be the $C^*$-algebra of all functions $S_n\to\C$ and let $f_{ij}\in\C^{S_n}$ be defined as
	\begin{equation*}
		f_{ij}(\pi):=\begin{cases}
		\begin{array}{cl}
		1 & \colon\pi(i)=j\\
		0 & \colon\textup{else}
		\end{array}
		\end{cases}
	\end{equation*}
	Then the following are equivalent:
	\begin{enumerate}
		\item[(a)] $a$ is semiclassical
		\item[(b)] $a\in\operatorname{mconv}(\mathcal{CP}^{(n)})$
		\item[(c)] There is a positive unital $*$-linear map $\phi\colon\C^{S_n}\to\operatorname{Mat}_s(\C)$ with $\phi(f_{ij})=a_{ij}$.
	\end{enumerate}
	\item[(ii)] If
	\begin{equation*}
		\sum_{k=1}^{n} a_{k,\pi(k)}\geqslant \frac{n-2}{n-1}\cdot I_s
	\end{equation*}
	for all $\pi\in S_n$, then $a$ is semiclassical.
\end{enumerate}
\end{theorem}
\begin{proof}
(i): For ``(a)$\Rightarrow$(b)" write 
$$a=\sum_\pi P_\pi\otimes q_\pi$$ 
with permutation matrices  $P_\pi=\left( p_{ij}^\pi\right)_{i,j}$  and  $q_\pi=v_\pi^*v_\pi$ for certain  $v_\pi\in{\rm Mat}_s(\C)$. Set 
\begin{align*}
u_{ij}&:={\rm diag}\left( p_{ij}^\pi I_s\mid \pi\in S_n \right)\in{\rm Mat}_{n!s}(\C) \\  
u&:=\left(u_{ij}\right)_{i,j}\in{\rm Mat}_n\left({\rm Mat}_{n!s}(\C)\right)\\ 
v&:=\left (v_\pi\right)_{\pi\in S_n}\in{\rm Mat}_{n!s,s}(\C).
\end{align*} 
Then $v$ is easily checked to be an isometry (using that rows and columns of $a$ sum to $I_s$), $u$ to be a quantum permutation matrix with commuting entries, and $v^*u_{ij} v=a_{ij}$ for all $i,j$.

For ``(b)$\Rightarrow$(c)" let $(u_{ij})\in\mathcal{CP}_t^{(n)}$ and $v\in\operatorname{Mat}_{t,s}(\C)$ such that $v^*v=I_s$ and
\begin{equation*}
	a_{ij}=v^*u_{ij}v.
\end{equation*}
Now we define the map
\begin{equation*}
	\psi\colon\C^{S_n}\to\operatorname{Mat}_t(\C),\quad f\mapsto\sum_{\pi\in S_n} f(\pi)\cdot\prod_{k=1}^{n} u_{k,\pi(k)}.
\end{equation*}
We want to show that $\psi$ is a unital $*$-algebra-homomorphism with $\psi(f_{ij})=u_{ij}$. The $*$-linearity is easy to see. Since the $u_{ij}$ commute we have
\begin{equation*}
	\psi(f)\psi(g)=\sum_{\pi,\sigma\in S_n} f(\pi)g(\sigma)\cdot\prod_{k=1}^{n}u_{k,\pi(k)}u_{k,\sigma(k)}.
\end{equation*}
Now if $\pi\neq\sigma$, then there is an index $k$ with $\pi(k)\neq\sigma(k)$. Since $(u_{ij})$ is a quantum permutation matrix, we get from \cref{rem:squares}~(i) that $u_{k,\pi(k)}u_{k,\sigma(k)}=0$. Thus
\begin{equation*}
	\psi(f)\psi(g)=\sum_{\pi\in S_n} f(\pi)g(\pi)\cdot\prod_{k=1}^{n}u_{k,\pi(k)}^2=\sum_{\pi\in S_n} (f\cdot g)(\pi)\cdot\prod_{k=1}^{n} u_{k,\pi(k)}=\psi(f\cdot g).
\end{equation*}
In order to compute $\psi(f_{ij})$ we first note that the sum in the definition of $\psi$ does not change when we extend it to non-injective maps. This follows from a similar argument as above. Now for $[n]:=\{1,\dots,n\}$ we have
\begin{align*}
	\psi(f_{nn}) &= \sum_{\pi\in S_n} f_{nn}(\pi)\cdot\prod_{k=1}^{n}u_{k,\pi(k)}= \sum_{\pi\colon[n]\to[n],\pi(n)=n}\prod_{k=1}^{n}u_{k,\pi(k)}\\
	&= u_{nn}\sum_{\pi\colon[n-1]\to[n]}\prod_{k=1}^{n-1}u_{k,\pi(k)}= u_{nn}\cdot\sum_{\tau_1=1}^{n}\cdots\sum_{\tau_{n-1}=1}^{n}\prod_{k=1}^{n-1}u_{k,\tau_k}\\
	&= u_{nn}\cdot\prod_{k=1}^{n-1}\sum_{\tau=1}^{n}u_{k,\tau}= u_{nn}
\end{align*}
The last equality follows from the fact that the columns of $(u_{ij})$ sum to $I_t$. Similarly we can compute $\psi(f_{ij})=u_{ij}$ for all $i,j$. Finally, $\psi$ is unital because both $(f_{ij})$ and $(u_{ij})$ are quantum permutation matrices. Now the map $\phi(f):=v^*\psi(f)v$ satisfies the requirements of (c).

For ``(c)$\Rightarrow$(a)" let $\delta_{\pi}\in\C^{S_n}$ be the characteristic function of $\pi$, i.e.\ the function that takes the value $1$ on $\pi$ and $0$ on all other permutations. From $\delta_{\pi}\geqslant 0$ and the properties of $\phi$ we obtain $q_{\pi}:=\phi(\delta_{\pi})\geqslant 0$ and $\sum_\pi q_{\pi}=I_s$. Now we compute
\begin{align*}
	a &= \sum_{i,j=1}^{n} E_{ij}\otimes a_{ij} = (\operatorname{id}\otimes\phi)\left(\sum_{i,j=1}^{n} E_{ij}\otimes f_{ij}\right)\\
	&= (\operatorname{id}\otimes\phi)\left(\sum_{\pi\in S_n} P_{\pi}\otimes\delta_{\pi}\right)=\sum_{\pi\in S_n} P_{\pi}\otimes q_{\pi}.
\end{align*}
Hence, $a$ is semiclassical.

(ii): Under the given assumption a straightforward computation shows that
\begin{equation*}
	\phi\colon\C^{S_n}\to\operatorname{Mat}_s(\C),\quad \delta_{\pi}\mapsto \frac{1}{(n-2)!n}\left(\left(\sum_{k=1}^n a_{k, \pi(k)}\right) - \frac{n-2}{n-1}\cdot I_s\right)
\end{equation*}
is a positive unital $*$-linear map with $\phi(f_{ij})=a_{ij}$. Thus, the claim follows from (i).
\end{proof}

\begin{remark}
(i) Note that the condition  in \cref{thm:nodil} (ii) is fulfilled by the constant quantum magic square $\left (\frac{1}{n}I_s\right)_{i,j}$, which is a relative interior point of $\mc M^{(n)}_s$. So \cref{thm:nodil} (ii) can be understood as a lower bound on the diameter of the set of semiclassical quantum magic squares inside $\mc M^{(n)}$.

(ii) Also note that \cref{thm:nodil} (i) together with \cref{rem:semiclassical} (ii) yields an efficient way for checking membership in $\operatorname{mconv}\left(\mc{CP}^{(n)}\right)$.
\end{remark}

Finally, in order to obtain the desired corollary, we need the following preparatory lemma.

\begin{lemma}\label{lem:split}
Let $u\in{\rm Mat}_s(\C)$ be a projector, $0\leqslant w \leqslant I_t \in{\rm Mat}_t(\C)$   (i.e.\ $w$ is a positive semidefinite contraction)  
and $v\in{\rm Mat}_{t,s}(\C)$ an isometry with $$u=v^*wv.$$ Then  $w$ is the direct sum of $u$ and some contraction $0\leqslant p\leqslant I_{t-s}$ with respect to the decomposition  $\C^t= v\C^s \oplus (v\C^s)^\perp.$ 
\end{lemma}
\begin{proof} With respect to the given decomposition  we write $$w=\left(\begin{array}{cc}u & r \\r^* & p\end{array}\right)$$ and compute \begin{align*}\left(\begin{array}{cc}u & r \\r^* & p\end{array}\right)&=w\geqslant w^2 =\left(\begin{array}{cc}u^2+rr^* & ur+rp \\r^*u+pr^* & r^*r+p^2\end{array}\right) \\ &=\left(\begin{array}{cc}u+rr^* & ur+rp \\r^*u+pr^* & r^*r+p^2\end{array}\right). \end{align*} We obtain $r=0$ and $p\geqslant p^2.$
\end{proof}

\begin{corollary}\label{cor:nodil}
(i) For every $n\geq 3$ and $s\geq 2$ there exist quantum magic squares in  $\mc M^{(n)}_s$ which are not semiclassical. In particular $${\rm mconv}\left(\mc{CP}^{(n)}\right) \subsetneq \mc M^{(n)}.$$

(ii) Every quantum permutation matrix of size $n\leq 3$ is semiclassical, 
but for every $n\geq 4, s\geq 2$ there exist quantum permutation matrices which are not semiclassical.
\end{corollary}
\begin{proof}
(i) From \cref{rem:mos} (iii) we know that $$\mc S_s^{(n)}\subsetneq \mc L_s^{(n)}$$ is a strict inclusion of closed convex cones with non-empty interior. So there must be an element $A\in\mc L_s^{(n)} \setminus \mc S_s^{(n)}$ for which the row and  
column sum is a positive definite  matrix $0<a\in{\rm Mat}_s(\C)^+$. 
Thus, there exists an invertible matrix $w\in{\rm Mat}_s(\C)$ such that $w^*aw=I_s$. Replacing the entries $a_{ij}$ of $A$ by $w^*a_{ij}w$ leads to a quantum magic square that is clearly not contained in $\mc S^{(n)}_s,$ meaning that it is not semiclassical.

(ii) For $n\leq 3$ we have $\mc{CP}^{(n)}=\mc P^{(n)}$, from which the first statement follows. For $n\geq 4$ and $s\geq 2$ there exist quantum permutation matrices  $U=\left(u_{ij}\right)_{i,j}\in\mc P_s^{(n)}\setminus \mc{CP}_s^{(n)}.$ One can take for example a  block-diagonal sum of two suitable smaller quantum permutation matrices.  If such $U$ was semiclassical, it would admit a dilation as in \cref{thm:nodil} (i), i.e.\ $u_{ij}=v^*w_{ij}v$ with  $\left(w_{ij}\right)_{i,j}\in\mc{CP}^{(n)}$. By \cref{lem:split} we can assume $$w_{ij}=\left(\begin{array}{cc}u_{ij} & 0 \\0 & p_{ij}\end{array}\right)$$ which contradicts the fact that the $u_{ij}$ do not all commute.
\end{proof}

\subsection{The matrix convex hull of quantum permutation matrices}
\label{ssec:bad}
Our main goal in this section is to prove the following result.

\begin{theorem}\label{thm:mconv}
For every $n\geq 3$ we have ${\rm mconv}\left(\mc P^{(n)}\right)\subsetneq \mc M^{(n)}.$ The difference already appears at level $s=2.$
\end{theorem}

Note that for $n=3$ the result follows from \cref{cor:nodil} (i) and the fact that $\mc{CP}^{(3)}=\mc P^{(3)}$. Proving the result for general $n\geq 4$ requires some more work, since we were not able to deduce it from an obvious embedding argument. 
We will instead establish a necessary condition for being in the matrix convex hull of $\mc P^{(n)}$ (\cref{prop:mconvconstraint}),  
show its failure for a quantum magic square from $\mc M^{(3)}_2$, 
and see that this implies failure for obvious embeddings 
 into quantum magic squares of larger size $n$.

\begin{construction}\label{con}
Let $A=(a_{ij})\in\mc{M}_s^{(n)}$. We define matrices $\operatorname{col}(A)$ and $\operatorname{diag}(A)$ by writing the entries of $A$ in one column and on the diagonal of a matrix respectively, where pairs of indices are ordered lexicographically, i.e. in the order $(1,1),(1,2),(1,3),\dots,(2,1),(2,2),\dots,(n,n)$. We can consider these matrices as elements of threefold tensor products as follows:
\begin{align*}
	\operatorname{col}(A) &:= \sum_{i,j=1}^{n} e_i\otimes e_j\otimes a_{ij}\in\C^n\otimes\C^n\otimes\Her_s(\C)\\
	\operatorname{diag}(A) &:= \sum_{i,j=1}^{n} E_{ii}\otimes E_{jj}\otimes a_{ij}\in{\rm Mat}_n(\C)\otimes{\rm Mat}_n(\C)\otimes\Her_s(\C)
\end{align*}
where $e_i$ and $E_{ij}$ denote the canonical basis of $\C^n$ and ${\rm Mat}_n(\C)$, respectively. Now let
\begin{equation*}
	\varphi(A):=\operatorname{diag}(A)-\operatorname{col}(A)\operatorname{col}(A)^*\in({\rm Mat}_n(\C)\otimes{\rm Mat}_n(\C)\otimes{\rm Mat}_s(\C))_{\operatorname{her}}.
\end{equation*}
For example for $n=2$ the matrix $\varphi(A)$ looks as follows:
\begin{equation*}
	\begin{pmatrix}
	a_{11}-a_{11}^2 & -a_{11}a_{12} & -a_{11}a_{21} & -a_{11}a_{22}\\
	-a_{12}a_{11} & a_{12}-a_{12}^2 & -a_{12}a_{21} & -a_{12}a_{22}\\
	-a_{21}a_{11} & -a_{21}a_{12} & a_{21}-a_{21}^2 & -a_{21}a_{22}\\
	-a_{22}a_{11} & -a_{22}a_{12} & -a_{22}a_{21} & a_{22}-a_{22}^2 
	\end{pmatrix}
\end{equation*}
For $n\geq 3$ we further define the matrix
\begin{equation*}
	\psi(A):=\sum_{\substack{{i,j,k,l=1}\\{i\neq j,k\neq l}}}^{n} E_{ij}\otimes E_{kl}\otimes\left(-\alpha_nI_s+\beta_na_{ik}+\beta_na_{jl}+\gamma_na_{il}+\gamma_na_{jk}\right)
\end{equation*}
where
\begin{equation*}
	\alpha_n:=\frac{1}{(n-1)(n-2)},\quad \beta_n:=\frac{n-1}{n(n-2)},\quad \gamma_n:=\frac{1}{n(n-2)}.
\end{equation*}
\demo\end{construction}

We let  $\mathcal{Z}^{(n)}$ denote the vector space of all matrices in ${\rm Mat}_n(\C)$ with zeros on the diagonal. Moreover, for $e:=(1,\dots,1)^t\in\C^n$ we define the following subspace of $\mathcal{Z}^{(n)}$:
\begin{equation*}
	\mathcal{Z}_e^{(n)}:=\{z\in\mathcal{Z}^{(n)}\,\vert\,ze=z^*e=0\}.
\end{equation*}

\begin{proposition}
	\label{prop:mconvconstraint}
	\textup{(i)} For all $A\in\mathcal{M}_s^{(n)}$ the following are equivalent:
	\begin{align}
		\exists X\in(\mathcal{Z}^{(n)}\otimes\mathcal{Z}^{(n)}\otimes\operatorname{Mat}_s(\C))_{\operatorname{her}}\colon &\quad \varphi(A)+X\geqslant 0 \label{eq:mconvconstraint}\\
		\exists X\in(\mathcal{Z}_e^{(n)}\otimes\mathcal{Z}_e^{(n)}\otimes\operatorname{Mat}_s(\C))_{\operatorname{her}}\colon &\quad \varphi(A)+\psi(A)+X\geqslant 0 \label{eq:mconvconstraintstrong}
	\end{align}
	\textup{(ii)} For $A\in\operatorname{mconv}\left(\mathcal{P}^{(n)}\right)_s$ the formulas \eqref{eq:mconvconstraint} and \eqref{eq:mconvconstraintstrong} are true.
\end{proposition}
\begin{proof}
	(i) By definition we have $\psi(A)\in(\mathcal{Z}^{(n)}\otimes\mathcal{Z}^{(n)}\otimes\operatorname{Mat}_s(\C))_{\operatorname{her}}$. Thus, the implication \eqref{eq:mconvconstraintstrong}$\Rightarrow$\eqref{eq:mconvconstraint} is obvious. For the same reason, \eqref{eq:mconvconstraint} implies the formula
	\begin{equation*}
		\exists X\in(\mathcal{Z}^{(n)}\otimes\mathcal{Z}^{(n)}\otimes\operatorname{Mat}_s(\C))_{\operatorname{her}}\colon\quad \varphi(A)+\psi(A)+X\geqslant 0.
	\end{equation*}
	Hence, for the reverse implication it suffices to show that if $X$ satisfies this formula, then $X\in(\mathcal{Z}_e^{(n)}\otimes\mathcal{Z}_e^{(n)}\otimes\operatorname{Mat}_s(\C))_{\operatorname{her}}$. Using the fact that $A\in\mathcal{M}_s^{(n)}$ one can easily verify that
	\begin{equation*}
		0 = (\varphi(A)+\psi(A))(e\otimes e_i\otimes I_s)
	\end{equation*}
	for $i=1,\dots,n$ ($\psi(A)$ is precisely designed for this to be true). Moreover, from $X\in(\mathcal{Z}^{(n)}\otimes\mathcal{Z}^{(n)}\otimes\operatorname{Mat}_s(\C))_{\operatorname{her}}$ it follows that
	\begin{equation*}
		0=(e\otimes e_i\otimes I_s)^*X(e\otimes e_i\otimes I_s).
	\end{equation*}
	Together we obtain
	\begin{equation*}
		0=(e\otimes e_i\otimes I_s)^*(\varphi(A)+\psi(A)+X)(e\otimes e_i\otimes I_s).
	\end{equation*}
	From this and $\varphi(A)+\psi(A)+X\geqslant 0$ we conclude, for example with the Cauchy--Schwarz inequality, that
	\begin{equation*}
		0=(\varphi(A)+\psi(A)+X)(e\otimes e_i\otimes I_s)=X(e\otimes e_i\otimes I_s).
	\end{equation*}
	Since $X$ is Hermitian, the same is true if we multiply from the other side, and by symmetry it is also true if we swap $e$ and $e_i$. Therefore we obtain $X\in(\mathcal{Z}_e^{(n)}\otimes\mathcal{Z}_e^{(n)}\otimes\operatorname{Mat}_s(\C))_{\operatorname{her}}$.
	
	(ii) Let $A\in\operatorname{mconv}\left(\mathcal{P}^{(n)}\right)_s$. By \cref{rem:mconv} (ii) there is a quantum permutation  matrix $U=(u_{ij})\in\mathcal{P}_t^{(n)}$ and an isometry $v\in{\rm Mat}_{t,s}(\C)$ such that
	\begin{equation*}
		a_{ij}=v^*u_{ij}v
	\end{equation*}
	for all $i,j$. By conjugating the $u_{ij}$ with a suitable unitary matrix if necessary, we may assume $v$ to be of the form
	\begin{equation*}
		v=\begin{pmatrix}
		I_s\\
		0
		\end{pmatrix} .
	\end{equation*}
	Then the $u_{ij}$ have the form
	\begin{equation*}
		u_{ij}=\begin{pmatrix}
		a_{ij} & b_{ij}^*\\
		b_{ij} & c_{ij}
		\end{pmatrix}
	\end{equation*}
	for some matrices $b_{ij}\in{\rm Mat}_{t-s,s}(\C)$ and $c_{ij}\in\Her_{t-s}(\C)$. Since the $u_{ij}$ are projectors, we have $u_{ij}^2=u_{ij}$ and from \cref{rem:squares} (i) we additionally know $u_{ij}u_{ik}=0=u_{ji}u_{ki}$ for $j\neq k$. This implies
	\begin{equation}
		\label{eq:b*brelations}
		\begin{split}
		b_{ij}^*b_{ij} &= a_{ij}-a_{ij}^2\\
		b_{ik}^*b_{ij} &= -a_{ik}a_{ij}\\
		b_{ji}^*b_{ki} &= -a_{ji}a_{ki}.
		\end{split}
	\end{equation}
	From the matrices $b_{ij}^*$ we construct a matrix $B$ in the same way we constructed $\operatorname{col}(A)$ from the $a_{ij}$, i.e.\ $B$ is just the column containing the $b_{ij}^*$. Now let
	\begin{equation*}
		X:=BB^*-\varphi(A).
	\end{equation*}
	Then we have
	\begin{equation*}
		\varphi(A)+X=BB^*\geqslant 0
	\end{equation*}
	and $X\in(\mathcal{Z}^{(n)}\otimes\mathcal{Z}^{(n)}\otimes{\rm Mat}_s(\C))_{\operatorname{her}}$ follows from \eqref{eq:b*brelations}.
\end{proof}

\begin{proof}[Proof of \cref{thm:mconv}]
	We show by induction on $n$ that \eqref{eq:mconvconstraint} cannot hold for all elements of $\mathcal{M}_2^{(n)}$. For $n=3$ let
	\begin{align*}
		a_{11} &:= \frac{1}{3}I_2+\frac{9}{62}\left(
		\begin{array}{cc}
		-\frac{34}{93} & \frac{4}{5}+\frac{2 i}{13} \\
		\frac{4}{5}-\frac{2 i}{13} & \frac{7}{16} \\
		\end{array}
		\right)\\
		a_{12} &:= \frac{1}{3}I_2+\frac{9}{62}\left(
		\begin{array}{cc}
		\frac{5}{6} & \frac{1}{3}-\frac{20 i}{81} \\
		\frac{1}{3}+\frac{20 i}{81} & -\frac{41}{55} \\
		\end{array}
		\right)\\
		a_{21} &:= \frac{1}{3}I_2+\frac{9}{62}\left(
		\begin{array}{cc}
		-\frac{2}{3} & -\frac{25}{92}-\frac{3 i}{7} \\
		-\frac{25}{92}+\frac{3 i}{7} & \frac{1}{34} \\
		\end{array}
		\right)\\
		a_{22} &:= \frac{1}{3}I_2+\frac{9}{62}\left(
		\begin{array}{cc}
		\frac{29}{30} & \frac{6}{35}-i \\
		\frac{6}{35}+i & -\frac{5}{8} \\
		\end{array}
		\right)
	\end{align*}
	and define $a_{13},a_{23},a_{31},a_{32},a_{33}$ such that $A=(a_{ij})$ is an element of $\mathcal{M}_2^{(3)}$. By \cref{prop:mconvconstraint}~(i) it is enough to show that this specific $A$ does not satisfy \eqref{eq:mconvconstraintstrong}. Note that the space $\mathcal{Z}_e^{(3)}$ is one-dimensional with Hermitian generator
	\begin{equation*}
	g:=\begin{pmatrix}
	0 & i & -i\\
	-i & 0 & i\\
	i & -i & 0
	\end{pmatrix}.
	\end{equation*}
	Thus, it remains to be shown that the formula
	\begin{equation*}
	\exists x\in\Her_2(\C)\colon\quad \varphi(A)+\psi(A)+g\otimes g\otimes x  \geqslant 0
	\end{equation*}
	is false. Now let
	\begin{align*}
	B_0:=\varphi(A)+\psi(A),\quad& B_1:=g\otimes g\otimes \begin{pmatrix}
	1 & 0\\
	0 & 0
	\end{pmatrix},\quad B_2:=g\otimes g\otimes \begin{pmatrix}
	0 & 1\\
	1 & 0
	\end{pmatrix},\\
	& B_3:=g\otimes g\otimes \begin{pmatrix}
	0 & -i\\
	i & 0
	\end{pmatrix},\quad B_4:=g\otimes g\otimes \begin{pmatrix}
	0 & 0\\
	0 & 1
	\end{pmatrix}.
	\end{align*}
	Then it suffices to find a matrix $Y\in(\mathrm{Mat}_3(\C)\otimes\mathrm{Mat}_3(\C)\otimes\mathrm{Mat}_2(\C))_{\operatorname{her}}$ such that
	\begin{align*}
	Y &\geqslant 0\\
	\operatorname{tr}(YB_0) &< 0\\
	\operatorname{tr}(YB_j) &= 0\quad\quad j=1,2,3,4.
	\end{align*}
	Such a matrix exists, which can be verified by the attached Mathematica\texttrademark{} file \emph{Counterexample.nb}. Note that all computations are over $\mathbb Q[i]$ and thus exact.

	Now let $n>3$. By the induction hypothesis there is an element $A\in\mathcal{M}_2^{(n-1)}$ that does not satisfy \eqref{eq:mconvconstraint}. 
	We claim that the same is true for
	\begin{equation*}
		A^{\prime}=\begin{pmatrix}
		A & 0\\
		0 & I_2
		\end{pmatrix}\in\mathcal{M}_2^{(n)}.
	\end{equation*}
	Let
	\begin{equation*}
		v=\begin{pmatrix}
		I_{n-1}\\
		0
		\end{pmatrix}\in{\rm Mat}_{n,n-1}(\C)
	\end{equation*}
	and let $X^{\prime}\in(\mathcal{Z}^{(n)}\otimes\mathcal{Z}^{(n)}\otimes{\rm Mat}_2(\C))_{\operatorname{her}}$ be an arbitrary element. Then we have
	\begin{equation*}
		X:=(v\otimes v\otimes I_2)^*X^{\prime}(v\otimes v\otimes I_2)\in(\mathcal{Z}^{(n-1)}\otimes\mathcal{Z}^{(n-1)}\otimes{\rm Mat}_2(\C))_{\operatorname{her}}
	\end{equation*}
	and
	\begin{equation*}
		(v\otimes v\otimes I_2)^*(\varphi(A^{\prime})+X^{\prime})(v\otimes v\otimes I_2)=\varphi(A)+X\not\geqslant 0.
	\end{equation*}
	It follows that
	\begin{equation*}
		\varphi(A^{\prime})+X^{\prime}\not\geqslant 0.
	\end{equation*}
	Since $X^{\prime}$ was arbitrary, the claim follows.
\end{proof}

In \cref{app:alternative} we provide an alternative proof of \cref{thm:mconv} which does not use this explicit example, but gives no information on the level at which the difference appears. 

\begin{definition}\label{def:arveson} Let $R=\bigcup_{s\geq 1} R_s$ be a matrix convex set (as in  \cref{def:mconv}). Then
$A =(a_{ij})_{i,j}\in R_s$ is an \emph{Arveson extreme point} of $R$ if whenever $A$ dilates to $B=(b_{ij})_{i,j}\in R_t$, 
then each $a_{ij}$ is a direct summand of $b_{ij} $ up to a simultaneous unitary conjugation. 
\end{definition}

In words, this means that $A$ only admits ``trivial" dilations, that is, of the form $$b_{ij} = \left(\begin{array}{cc}a_{ij} & 0 \\0 & c_{ij}\end{array}\right),$$  see also \cite{Eve19} for more details. 

\begin{corollary}\label{cor:arv}
Every quantum permutation matrix in $\mc P^{(n)}$ is an Arveson extreme point 
of the matrix convex set $\mc M^{(n)}$ of all quantum  magic squares. 
For $n\geq 3$ not every Arveson extreme point of $\mc M^{(n)}$ is a quantum permutation matrix.
\end{corollary}

\begin{proof} Let $U=\left(u_{ij}\right)_{i,j}\in\mc P^{(n)}$ and assume
$U=\left(v^*a_{ij}v\right)_{i,j}$ for some quantum magic square $A=\left(a_{ij}\right)_{i,j}\in\mc M^{(n)}$ and an isometry $v$. Since $A$ is a quantum magic square  we have $0\leqslant a_{ij}\leqslant I$  for all $i,j$. By \cref{lem:split}, each $a_{ij}$ splits off $u_{ij}$ as a direct summand, and thus $U$ is an Arveson extreme point.

Note that $\mc M^{(n)}$ is a compact free spectrahedron, i.e.\ definable by a linear matrix inequality (even by finitely many linear inequalities with scalar coefficients). In  \cite{Eve19} it is shown that the matrix convex hull of the absolute extreme points, and in particular of the Arveson extreme points, is the full spectrahedron. So in view of \cref{thm:mconv} not every Arveson extreme point (not even every absolute extreme point) can be a quantum permutation matrix. 
\end{proof}

\appendix

\section{Alternative proof of Theorem \ref{thm:mconv}} 
\label{app:alternative}

In the proof of \cref{thm:mconv} we gave an explicit element $A\in\mathcal{M}_s^{(3)}$ that does not satisfy the constraint \eqref{eq:mconvconstraint} from \cref{prop:mconvconstraint}. In this appendix we give a conceptual proof that such an element $A$ exists. However, in contrast to the original proof this does not show that we can choose $s=2$.
Recall from \cref{con} the definition of $\varphi(A)$ and $\psi(A)$. 

\begin{proposition}
	\label{prop:mconvweak}
	There exists a natural number $s$ and an element $A\in\mathcal{M}_s^{(3)}$ such that
	\begin{equation*}
		\forall X\in(\mathcal{Z}^{(3)}\otimes\mathcal{Z}^{(3)}\otimes{\rm Mat}_s(\C))_{\operatorname{her}}\colon\quad \varphi(A)+X\not\geqslant 0.
	\end{equation*}
\end{proposition}
\begin{proof}
	By \cref{cor:arv} there is an Arveson extreme point $A\in\mathcal{M}_s^{(3)}$ of $\mathcal{M}^{(3)}$ that is not a quantum permutation matrix. Note that in the case $n=3$ this does not involve \cref{thm:mconv}  but already follows from \cref{cor:nodil}.
Without loss of generality we may assume that
	\begin{equation*}
		a_{11}-a_{11}^2\neq 0.
	\end{equation*}
	We show by contradiction that $A$ satisfies the claim. So assume that there exists an $X\in(\mathcal{Z}^{(3)}\otimes\mathcal{Z}^{(3)}\otimes{\rm Mat}_s(\C))_{\operatorname{her}}$ such that
	\begin{equation*}
		\varphi(A)+X\geqslant 0.
	\end{equation*}
	Then we can find a matrix $B$ such that
	\begin{equation*}
		\varphi(A)+X=BB^*.
	\end{equation*}
	We write $B$ as
	\begin{equation*}
		B=\sum_{i,j=1}^{3} e_i\otimes e_j\otimes b_{ij}^*
	\end{equation*}
	with matrices $b_{ij}\in\mathrm{Mat}_{t,s}(\C)$ for some natural number $t$ (cf.\ proof of \cref{prop:mconvconstraint}). It follows that
	\begin{equation*}
		b_{11}^*b_{11}=a_{11}-a_{11}^2\neq 0.
	\end{equation*}
	In particular, we can find a vector $v\in\C^s$ such that
	\begin{equation*}
		v^*b_{11}^*b_{11}^{}v=1.
	\end{equation*}
	On the other hand we obtain from $A\in\mathcal{M}^{(3)}$ that
	\begin{align*}
		(b_{i1}+b_{i2}+b_{i3})^*(b_{i1}+b_{i2}+b_{i3}) &= \sum_{k,l=1}^{3}b_{ik}^*b_{il}
		= \sum_{k,l=1}^{3}(\delta_{kl}a_{ik}-a_{ik}a_{il})\\
		&= \sum_{k=1}^{3}(a_{ik}-a_{ik}\cdot I_s)
		= 0.
	\end{align*}
	This implies
	\begin{equation}
		\label{eq:bsums}
		b_{i1}+b_{i2}+b_{i3}=b_{1j}+b_{2j}+b_{3j}=0
	\end{equation}
	for all $i,j=1,2,3$. For $i,j=1,2,3$ let $p_{ij}$ be the Moore--Penrose inverse of $a_{ij}^{\frac{1}{2}}$. Since the $a_{ij}^{\frac{1}{2}}$ are Hermitian, it is well known that $p_{ij}$ is Hermitian as well and that
	\begin{equation}
		\label{eq:pa=ap}
		p_{ij}a_{ij}^{\frac{1}{2}}=a_{ij}^{\frac{1}{2}}p_{ij}
	\end{equation}
	is the orthogonal projection onto $\ker(a_{ij})^{\perp}$. Since $b_{ij}^*b_{ij}=a_{ij}-a_{ij}^2$, this gives
	\begin{equation}
		\label{eq:bpa=b}
		b_{ij}p_{ij}a_{ij}^{\frac{1}{2}}=b_{ij}.
	\end{equation}
	
	Our goal is now to show that there are real numbers $c_{ij}\geq 0$ such that the matrix $A^{\prime}=(a_{ij}^{\prime})$ with
	\begin{equation*}
		a_{ij}^{\prime}:=\begin{pmatrix}
		a_{ij} & b_{ij}^*b_{11}v\\
		v^*b_{11}^*b_{ij} & v^*b_{11}^*b_{ij}p_{ij}^2b_{ij}^*b_{11}v+c_{ij}
		\end{pmatrix}\in{\rm Her}_{s+1}(\C)
	\end{equation*}
	is an element of $\mathcal{M}_{s+1}^{(3)}$. If that is the case, then $A$ cannot be an Arveson extreme point of $\mathcal{M}^{(3)}$ since $b_{11}^*b_{11}^{}v\neq 0$ by our construction. So we arrive at a contradiction. Thus, it remains to be shown that $A^{\prime}\in\mathcal{M}_{s+1}^{(3)}$, i.e.\ that $a_{ij}^{\prime}\geqslant 0$ and that the sum of the $a_{ij}^{\prime}$ in each row and each column is $I_{s+1}$. For $a_{ij}^{\prime}\geqslant 0$ we compute
	\begin{align*}
		0 &\leqslant \begin{pmatrix}
		I_s\\
		v^*b_{11}^*b_{ij}p_{ij}^2
		\end{pmatrix}a_{ij}\begin{pmatrix}
		I_s & p_{ij}^2b_{ij}^*b_{11}v
		\end{pmatrix}+\begin{pmatrix}
		0 & 0\\
		0 & c_{ij}
		\end{pmatrix}\\
		&= \begin{pmatrix}
		a_{ij} & a_{ij}p_{ij}^2b_{ij}^*b_{11}v\\
		v^*b_{11}^*b_{ij}p_{ij}^2a_{ij} & v^*b_{11}^*b_{ij}p_{ij}^2a_{ij}p_{ij}^2b_{ij}^*b_{11}v+c_{ij}
		\end{pmatrix}\\
		&\stackrel{\eqref{eq:pa=ap}}{=} \begin{pmatrix}
		a_{ij} & (a_{ij}^{\frac{1}{2}}p_{ij})^2b_{ij}^*b_{11}v\\
		v^*b_{11}^*b_{ij}(p_{ij}a_{ij}^{\frac{1}{2}})^2 & v^*b_{11}^*b_{ij}p_{ij}a_{ij}^{\frac{1}{2}}p_{ij}^2a_{ij}^{\frac{1}{2}}p_{ij}b_{ij}^*b_{11}v+c_{ij}
		\end{pmatrix}\\
		&\stackrel{\eqref{eq:bpa=b}}{=} a_{ij}^{\prime}.
	\end{align*}
	From $A\in\mathcal{M}^{(3)}$ and \eqref{eq:bsums} we obtain the condition on the sums for all but the lower right entry. For the lower right entry we first compute
	\begin{align*}
		\left(\sum_{k=1}^{3}b_{ik}p_{ik}^2b_{ik}^*\right)^2 &= \sum_{k,l=1}^{3} b_{ik}p_{ik}^2b_{ik}^*b_{il}p_{il}^2b_{il}^*\\
		&= \sum_{k,l=1}^{3} b_{ik}p_{ik}^2(\delta_{kl}a_{ik}-a_{ik}a_{il})p_{il}^2b_{il}^*\\
		&\stackrel{\eqref{eq:pa=ap}}{=} \left(\sum_{k=1}^{3} b_{ik}p_{ik}a_{ik}^{\frac{1}{2}}p_{ik}^2a_{ik}^{\frac{1}{2}}p_{ik}b_{ik}^*\right)-\left(\sum_{k,l=1}^{3}b_{ik}(p_{ik}a_{ik}^{\frac{1}{2}})^2(a_{il}^{\frac{1}{2}}p_{il})^2b_{il}^*\right)\\
		&\stackrel{\eqref{eq:bpa=b}}{=} \left(\sum_{k=1}^{3} b_{ik}p_{ik}^2b_{ik}^*\right)-\left(\sum_{k,l=1}^{3}b_{ik}b_{il}^*\right) 
		\stackrel{\eqref{eq:bsums}}{=} \sum_{k=1}^{3} b_{ik}p_{ik}^2b_{ik}^*.
	\end{align*}
	It follows that
	\begin{equation*}
		\sum_{k=1}^{3} b_{ik}p_{ik}^2b_{ik}^*\leqslant I_t.
	\end{equation*}
	Since $v^*b_{11}^*b_{11}v=1$, we therefore get
	\begin{align*}
		s_{i*} &:= \sum_{j=1}^{3} v^*b_{11}^*b_{ij}p_{ij}^2b_{ij}^*b_{11}v\leq 1\\
		s_{*j} &:= \sum_{i=1}^{3} v^*b_{11}^*b_{ij}p_{ij}^2b_{ij}^*b_{11}v\leq 1\\
		s &:= \sum_{j=1}^{3} (1-s_{*j})=\sum_{i=1}^{3} (1-s_{i*})\geq 0.
	\end{align*}
	So if we choose
	\begin{equation*}
		c_{ij}:=\begin{cases}
		\begin{array}{cl}
		0 & \text{ , if $s=0$}\\
		\displaystyle\frac{(1-s_{i*})(1-s_{*j})}{s} & \text{ , else}
		\end{array}
		\end{cases}
	\end{equation*}
	then we have $c_{ij}\geq 0$ and the condition on the sums is satisfied.
\end{proof}

\section*{Acknowledgements}
TD and TN acknowledge funding through FWF project P 29496-N35 (free semialgebraic geometry and convexity). 
We thank the participants of the workshop ``Analytical and combinatorial aspects of quantum information theory'' (Edinburgh, September 2019) for stimulating discussions.

\newcommand{\etalchar}[1]{$^{#1}$}


\begin{thebibliography}{18}
\bibliographystyle{acm}
\bibitem{At19}
{\sc Atserias, A., Man\v{c}inska, L., Roberson, D.E., \v{S}\'{a}mal, R.,
  Severini, S., and Varvitsiotis, A.}
\newblock Quantum and non-signalling graph isomorphisms.
\newblock {\em J. Combin. Theory Ser. B 136,} 289--328 (2019).
\newblock \href{https://arxiv.org/abs/1611.09837}{arXiv:1611.09837 [quant-ph]}

\bibitem{Ba07}
{\sc Banica, T., Bichon, J., and Collins, B.}
\newblock Quantum permutation groups: a survey.
\newblock In {\em Noncommutative harmonic analysis with applications to
  probability}, vol.~78 of {\em Banach Center Publ.} Polish Acad. Sci. Inst.
  Math., Warsaw, 13--34 (2007).
\newblock \href{https://arxiv.org/abs/math/0612724}{	arXiv:math/0612724 [math.CO]}


\bibitem{Ba02}
{\sc Barvinok, A.}
\newblock A course in convexity, vol.~54 of {\em Graduate Studies in
  Mathematics}.
\newblock American Mathematical Society, Providence, RI (2002).

\bibitem{Bru87}
{\sc Brualdi, R.~A.}
\newblock Some applications of doubly stochastic matrices.
\newblock In {\em Proceedings of the {V}ictoria {C}onference on {C}ombinatorial
  {M}atrix {A}nalysis ({V}ictoria, {BC}, 1987)\/} vol.~107,
  77--100 (1988). 

\bibitem{Bud09}
{\sc Budish, E., Che, Y.-K., Kojima, F., and Milgrom, P.}
\newblock Implementing random assignments: A generalization of the Birkhoff-von
  Neumann theorem.
\newblock Tech. rep. (2009).

\bibitem{Choi90}
{\sc Choi, M.D., and Wu, P.Y.}
\newblock Convex combinations of projections.
\newblock {\em Linear Algebra Appl. 136}, 25--42  (1990).

\bibitem{Ef97}
{\sc Effros, E.G., and Winkler, S.}
\newblock Matrix convexity: operator analogues of the bipolar and
  {H}ahn-{B}anach theorems.
\newblock {\em J. Funct. Anal. 144} 1, 117--152 (1997).

\bibitem{Eve19}
{\sc Evert, E., and Helton, J.W.}
\newblock Arveson extreme points span free spectrahedra.
\newblock {\em Math. Ann. 375} 1-2, 629--653 (2019).
\newblock \href{https://arxiv.org/abs/1806.09053}{arXiv:1806.09053 [math.OA]}

\bibitem{Fil67}
{\sc Fillmore, P.A.}
\newblock Sums of operators with square zero.
\newblock {\em Acta Sci. Math. (Szeged) 28\/}, 285--288 (1967).

\bibitem{Fr16}
{\sc Fritz, T., Netzer, T., and Thom, A.}
\newblock {Spectrahedral containment and operator systems with
  finite-dimensional realization}.
 \newblock  {\em SIAM J. Appl. Algebra Geom. 1}, 556--574 (2017).
\newblock \href{https://arxiv.org/abs/1609.07908}{arXiv:1609.07908 [math.FA]}

\bibitem{Haa11}
{\sc Haagerup, U., and Musat, M.}
\newblock Factorization and dilation problems for completely positive maps on
  von {N}eumann algebras.
\newblock {\em Comm. Math. Phys. 303} 2, 555--594 (2011).
\newblock \href{https://arxiv.org/abs/1009.0778}{arXiv:1009.0778 [math.OA]}

\bibitem{Hu19}
{\sc Huber, B., and Netzer, T.}
\newblock A note on non-commutative polytopes and polyhedra.
\newblock to appear in {\em Adv. Geom.\/}  (2018).
\newblock \href{https://arxiv.org/abs/1809.00476}{arXiv:1809.00476 [math.AG]}

\bibitem{Lan93}
{\sc Landau, L.J., and Streater, R.F.}
\newblock On {B}irkhoff's theorem for doubly stochastic completely positive
  maps of matrix algebras.
\newblock {\em Linear Algebra Appl. 193\/}, 107--127  (1993).

\bibitem{Lu17}
{\sc Lupoini, M., Man\v{c}inska, L., and Roberson, D.E.}
\newblock Nonlocal games and quantum permutation groups. (2017).
\newblock \href{https://arxiv.org/abs/1712.01820}{arXiv:1712.01820 [quant-ph]} 

\bibitem{Me09b}
{\sc Mendl, C.B., and Wolf, M.M.}
\newblock {Unital quantum channels - convex structure and revivals of
  Birkhoff's theorem}.
\newblock {\em Communications in Mathematical Physics 289} 3,  1057--1086 (2009).
\newblock \href{https://arxiv.org/abs/0806.2820}{arXiv:0806.2820 [quant-ph]}

\bibitem{Ne19}
{\sc Netzer, T.}
\newblock {Free semialgebraic geometry}.
\newblock {\em Internat. Math. Nachrichten 240}, 31--41 (2019). 
\newblock 	\href{https://arxiv.org/abs/1902.11170}{arXiv:1902.11170 [math.AG]} 

\bibitem{Pau02}
{\sc Paulsen, V.}
\newblock {Completely bounded maps and operator algebras}, vol.~78 of {\em
  Cambridge Studies in Advanced Mathematics}.
\newblock Cambridge University Press, Cambridge (2002).

\bibitem{Pa19}
{\sc Paulsen, V., and Rahaman, M.}
\newblock Bisynchronous games and factorizable maps.  (2019).
\newblock \href{https://arxiv.org/abs/1908.03842}{arXiv:1908.03842 [quant-ph]}


\bibitem{Zie95}
{\sc Ziegler, G.M.}
\newblock {Lectures on polytopes}, vol.~152 of {\em Graduate Texts in
  Mathematics}.
\newblock Springer-Verlag, New York (1995).

\end{thebibliography}
\end{document}